\documentclass[conference]{IEEEtran}
\IEEEoverridecommandlockouts    

\usepackage{multirow}
\usepackage{lipsum}
\usepackage{amsmath}
\usepackage{amssymb}

\usepackage{amsthm}
\newtheorem{assumption}{Assumption}
\usepackage[ruled,vlined]{algorithm2e}
\usepackage{graphicx}
\usepackage{booktabs}
\usepackage{tabularx}
\usepackage{float}
\usepackage[compatibility=false]{caption}
\usepackage{subcaption}
\usepackage{xcolor}
\usepackage{mathtools}
\newtheorem{remark}{Remark}

\newtheorem{theorem}{Theorem}
\newtheorem{lemma}{Lemma}
\newtheorem{proposition}{Proposition}
\graphicspath{{./Figures/}}

\usepackage{balance}

\usepackage{xcolor}
\usepackage{amsmath}

\DeclareMathOperator*{\argmin}{arg\,min}

\newcommand{\ch}[1]{\textcolor{black}{#1}}

\newcommand{\sr}[1]{\textcolor{black}{#1}}

\title{\vspace{0.25in}\LARGE \bf On the Input Sensitivity of METANET Models and\\the Robustness of Dynamic Calibration}

\author{
	\parbox{\textwidth}{%
		\centering
		Cameron Hickert$^{*}$, Shreyaa Raghavan$^{*}$, Cathy Wu%
	}%
	\thanks{$^{*}$These authors contributed equally to this work.}%
    \thanks{The authors are with the Laboratory for Information \& Decision Systems and the Institute for Data, Systems, and Society, Massachusetts Institute of Technology, Cambridge, MA 02139, USA. C. Wu is also with the Dept. of Civil and Environmental Engineering. {(Emails: \tt\small shreyaar@mit.edu; chickert@mit.edu; cathywu@mit.edu})}%
    \thanks{\copyright~2026 IEEE. Personal use of this material is permitted. Personal use of this material is permitted. Permission must be obtained for all other uses, including reprinting/republishing this material for advertising or promotional purposes, collecting new collected works for resale or redistribution to servers or lists, or reuse of any copyrighted component of this work in other works.}%
}

\begin{document}
	
	\maketitle
	\thispagestyle{empty}
	\pagestyle{empty}
	
	\begin{abstract}
        While traffic modeling and control rely upon effective calibration of macroscopic traffic simulation models like METANET, recent empirical work shows these models can exhibit severe sensitivity to input noise. This paper explains this phenomenon through a string-stability analysis, demonstrating that a calibrated METANET model can amplify small additive perturbations to boundary conditions along the corridor, causing the simulated state to diverge from the nominal baseline. The input sensitivity in off-nominal cases may compromise the model's capabilities for counterfactual analysis, which is required for use cases of interest like the design of large-scale variable speed limits. To address these issues, this work demonstrates how a dynamic, time-varying approach to calibrate model parameters can achieve robustness and improved accuracy of the resulting macrosimulation. We show analytically that under stated regularity and perturbation assumptions, dynamic calibration achieves a tighter cost deviation bound than static calibration, and we validate this behavior in highway environments with synthetic and real-world data. Ultimately, this can enable more trustworthy traffic simulation and more effective evaluation of control strategies.
	\end{abstract}
	
\section{Introduction}
\label{sec:introduction}

Effective control of large-scale traffic networks relies on accurate, well-calibrated macroscopic simulations \cite{papageorgiou2019role}. By approximating vehicle movements as continuous fluid flows, models like METANET provide the computational tractability and scalability for large-scale traffic network control and simulation use cases. While METANET is widely utilized for traffic state estimation and control \cite{kotsialos2002traffic, mohammadian2021performance, messmer1990metanet}, the validity of any downstream control strategy hinges on robust and accurate calibration.

Although momentum-based models like METANET can reproduce physical phenomena such as stop-and-go waves, their sensitivity to boundary conditions like upstream flow or downstream density remains a challenge \cite{mohammadian2021performance}. \ch{Despite this,} standard static calibration is \ch{often} treated as a one-time estimation procedure on historical data, which risks overfitting and neglects stress testing. Consequently, \sr{minute perturbations, or additive input noise}, can undergo severe spatial amplification, \sr{causing the simulation to} irrecoverably diverge from the baseline \cite{raghavan2026dynamic}. This \sr{input sensitivity} \ch{can} drastically compromise the simulation's value for counterfactual evaluation and traffic control.

Recent empirical work demonstrates that shifting to time-varying, dynamic parameter calibration \ch{may yield} more robust METANET models \ch{for stop-and-go traffic phenomena} \cite{raghavan2026dynamic}. \ch{Our} paper \ch{proposes an analytical explanation for} why this is the case \ch{through} formal mathematical analysis regarding METANET's \ch{sensitivity to input perturbations}. By formulating calibration as a finite-horizon optimal control problem, we derive bounds demonstrating dynamic calibration's robustness and accuracy advantages and validate them numerically.

This work makes three primary contributions:
\begin{itemize}
    \item \textbf{Formalizing Input Sensitivity:} Using string stability analysis, we characterize METANET's sensitivity in the face of exogenous boundary condition perturbations, such as that due to sensor noise.
    \item \textbf{Analytical Bounding:} We derive an analytical bound on cost deviation from boundary noise, mathematically proving that dynamic calibration achieves a tighter upper bound. Furthermore, we show that the performance advantage of dynamic over static calibration scales with ground truth state variance.
    \item \textbf{Empirical Validation:} We validate these theoretical findings using numerical experiments on both synthetic scenarios and real-world I-24 MOTION trajectory data.
\end{itemize}

The remainder of this paper is organized as follows. The relevant literature is described in Section \ref{sec:related_work}. Section \ref{sec:prob_form} \ch{describes METANET and} formulates calibration as a finite-horizon optimal control problem. Section \ref{sec:string_instability} establishes METANET's \ch{input sensitivity} to data perturbations. Section \ref{sec:analysis} derives analytical bounds on dynamic calibration robustness. Section \ref{sec:experiments} presents experimental validation. Finally, Section \ref{sec:conclusion} concludes the paper.

\section{Related Work}
\label{sec:related_work}

Although numerous METANET variants have been proposed, they are largely evaluated in specific settings with fixed data inputs -- where parameters are prone to overfitting -- or on synthetic datasets where underlying parameters are known to be time-invariant. However, because driving behavior inherently varies over time, online and dynamic calibration methods have been explored in dynamic traffic assignment \cite{antoniou2009off} and for \ch{density-only} macroscopic models \cite{thonhofer2014online}. More recently, empirical studies on \ch{momentum-based} models \ch{(which couple density and velocity)} have revealed that a statically calibrated METANET may struggle with boundary variations, as minor input perturbations can quickly cascade into catastrophic instabilities within the simulated traffic state \cite{mohammadian2021performance, raghavan2026dynamic}. By recasting a one-time static estimation into a dynamic, rolling-horizon control problem, \cite{raghavan2026dynamic} demonstrated substantial empirical improvements in robustness and predictive accuracy. Our work seeks to rigorously explain the mechanisms driving this surprising robustness. By formalizing the \ch{sensitivity} of METANET to exogenous data perturbations, we fill a theoretical gap. This understanding operates in parallel with broader efforts to enhance static calibration, uniting these two lines of inquiry to inform the development of future, highly reliable calibration methods \cite{huang2025metanet}.

To achieve this, our work builds upon a rich history of formal stability analysis in traffic flow, shifting the focus from physical phenomena to calibration robustness. Previous efforts have utilized Lyapunov stability to analyze microscopic models \cite{palatella2013nonlinear} and macroscopic physical control strategies \cite{friedrich2023lyapunov}. Similarly, formal nonlinear dynamics and bifurcation analyses have been used to prove macroscopic instability driven by physical road features, such as varying slopes \cite{cen2023global}. In a particularly relevant study, \cite{yi2003stability} provides a nonlinear stability analysis of \ch{momentum-based} Payne-Whitham models, employing wavefront expansion to assess how models handle large, internal shocks. However, this work evaluates theoretical models under the assumption of fixed, known parameters and defines stability via internal spatial gradients. Rather than analyzing internal shocks, our work \ch{shows how exogenous boundary perturbations can catastrophically amplify across METANET segments}, and we provide analytical bounds to mathematically \ch{demonstrate how} dynamic calibration can better handle these perturbed scenarios over standard static approaches.

\section{Problem Formulation}
\label{sec:prob_form}

\subsection{Traffic Dynamics Model (METANET)}

Here, we describe the discrete-time macroscopic traffic dynamics model METANET \cite{messmer1990metanet}. We consider a highway divided into $N$ segments, each of length $L$, and a simulation horizon of $H$ with a simulation time step of $T$ seconds. The macroscopic traffic state of a segment $\ell$ is determined by its density $\rho_t(\ell)$ and average velocity $v_t(\ell)$ at time $t$. The aggregate network state is $x_t = [v_{t}(1), \ldots, v_{t}(N), \rho_{t}(1), \ldots,  \rho_{t}(N)]^\top \in \mathbb{R}^{2N}$.

The METANET model describes the evolution of the traffic state $x_t$ over time, subject to the initial state $x_0$ and boundary conditions $d_t \in \mathbb{R}^m$ -- such as on-ramp flows, upstream boundary flow, downstream density, etc., where $m$ is the number of inputs -- through the following equations:
\begin{align}
        v_{t+1}(\ell) &= v_{t}(\ell) + \frac{T}{\tau} (V[\rho_{t}(\ell)] - v_{t}(\ell))  \\
        &  \qquad \quad \hspace{0.4em}+ \frac{T}{L} \, v_{t}(\ell) \left( v_{t}(\ell-1)- v_{t}(\ell)\right) \notag \\
        & \qquad \quad \hspace{0.4em} - \frac{\eta T}{\tau L} \frac{\rho_{t}(\ell+1) - \rho_{t}(\ell)}{\rho_{t}(\ell) + \kappa} \notag, \\
        V[\rho_{t}(\ell)] &=  v_{free} \exp \left[-\left(\frac{\rho_{t}(\ell)}{a \, \rho_{cr}}\right)^{a}\right] \label{V},
\end{align}
\begin{align}
        q_{t}(\ell) &= \rho_{t}(\ell) v_{t}(\ell) \lambda_\ell \label{flow}, \\
                \rho_{t+1}(\ell) &= \rho_{t}(\ell) + \frac{T}{L \lambda_\ell} \left(q_{t}(\ell-1) - \frac{q_{t}(\ell)}{1 - \beta_{\ell}} + r_t(\ell) \right),
        \label{density}
\end{align}

where $\lambda_{\ell}$ is the number of lanes on segment $\ell$, $q_{t}(\ell)$ is segment $\ell$'s inflow at time $t$, $\beta_{\ell}$ is the off-ramp turning ratio at segment $\ell$, $r_t(\ell)$ is the on-ramp inflow for segment $\ell$ at time $t$, and $[\tau,\, \eta,\, \kappa,\, v_{free}, \, a,\, \rho_{\rm cr},]$  are the parameters that require calibration. The latter three are the parameters that describe the (exponential) fundamental diagram, as shown in equation \ref{V}. We write the full state transition as
\begin{equation}
  x_{t+1} = F(x_t,\, d_t;\, \theta),
  \label{eq:dynamics}
\end{equation}
where $\theta \in \mathcal{B} \subset \mathbb{R}^n$ is the vector of METANET parameters to be calibrated, i.e.,
$\theta = [\tau,\, \eta,\, \kappa,\, v_{free}, \, a,\, \rho_{\rm cr},]^\top$,
and $\mathcal{B}$ denotes the set of physically admissible parameter values.

We define a perturbation on the boundary conditions $\varepsilon_t \in \mathbb{R}^m$ that shifts the nominal conditions to the perturbed conditions $\tilde{d}_t$:
\begin{equation}
    \tilde{d}_t = d_t + \varepsilon_t, \qquad \|\varepsilon_t\| \leq \bar{\varepsilon} ,
\end{equation}
where the perturbations are drawn from the uncertainty set
$U = \{\varepsilon \mid \|\varepsilon_t\| \leq \bar{\varepsilon} \ \forall\, t\}$. Under the nominal conditions, the METANET state evolves as
\begin{equation}
    x_{t+1} = f(x_t,\, d_t\,;\, \theta_t),
\end{equation}
while under the perturbed conditions, the state evolves as
\begin{equation}
    \tilde{x}_{t+1} = f(\tilde{x}_t,\, \tilde{d}_t\,;\, \theta_t),
\end{equation}
where the same parameter $\theta_t$ is applied. Note that parameters are not \ch{updated in response} to the perturbation. We also define the state deviation at time $t$ as
\begin{equation}
    e_t(\theta_t, \varepsilon) := \tilde{x}_t(\theta_t, \varepsilon) - x_t(\theta_t)
\end{equation}
Here, we specify the state variables with the parameters $\theta_t$ and perturbation $\varepsilon_t$ to denote the traffic state that is generated by those conditions. For brevity, we later refer to  $e_t(\theta_t, \varepsilon)$ as $e_t$, $\tilde{x}_t(\theta_t, \varepsilon)$ as $\tilde{x}_t$, and $x_t(\theta_t)$ as $x_t$. Since both systems share the same initial condition $\tilde{x}_0 = x_0$,
we have $e_0 = 0$. 

\subsection{Calibration as a Finite-Horizon Optimal Control Problem}

The goal of calibrating METANET, or any traffic macrosimulation, is to find parameters that make the state trajectory closely match ground-truth traffic data. Following~\cite{raghavan2026dynamic}, we treat the parameters as time-varying and minimize the objective
\begin{align}
J(\theta_{0:H-1})
  &= \sum_{t=0}^{H} c_t\!\left(x_t(\theta_{0:t})\right) \\
  &= \sum_{t=0}^{H} \left\| x_t(\theta_{0:t}) - x^{\mathrm{gt}}_t \right\|^2,
\end{align}
where $\theta_{0:t}$ collects the parameters up to time $t$, $x_t(\theta_{0:t})$ is the state trajectory induced by $\theta_{0:t}$, and $c_t$ is the quadratic stage cost measuring the squared deviation from the ground-truth state $x^{\mathrm{gt}}_t$. Calibration is then the finite-horizon optimal control problem
\begin{equation}
\min_{\theta_{0:H-1}\in\mathcal{B}^{H}} J(\theta_{0:H-1})
\quad \text{s.t.} \quad \eqref{eq:dynamics},\ x_0 = x^{\mathrm{gt}}_0.
\label{eq:calibration-ocp}
\end{equation}

We write $\bar\theta \in \mathcal{B}$ for a \emph{static} parameter vector, in which the parameters are fixed to a single value over the entire horizon, and $\theta^{*}_{0:H-1}$ for an optimal solution of~\eqref{eq:calibration-ocp}. 

To quantify how much a static $\bar\theta$ loses at stage $t$ relative to acting optimally at the current state, we define the per-stage
\emph{advantage function}
\begin{equation}
A_t(x, \bar\theta)
  := Q_t(x, \bar\theta) - \min_{\theta \in \mathcal{B}} Q_t(x, \theta)
  \ \ge\ 0,
\label{eq:advantage}
\end{equation}
where $Q_t(x, \theta)$ denotes the standard $Q$-function: the total cost incurred by applying $\theta$ at stage $t$ from state $x$ and following an optimal continuation thereafter. Note that $\theta^*_t(x) \;=\; \arg\min_{\theta \in \mathcal{B}}\; Q_t(x,\, \theta)$, and equality in~\eqref{eq:advantage} holds if and only if $\bar \theta = \theta^*_t$.

\section{METANET's Input Sensitivity via String Stability Analysis}
\label{sec:string_instability}

To isolate and quantify \ch{divergence due to additive input noise} $\varepsilon$, we evaluate the system through the lens of discrete-time string stability. By treating the METANET segment update equations as a forced linear system, we decouple the unforced local dynamics \ch{(internal segment dynamics)} from the upstream boundary conditions. 
This is achieved by linearizing the \ch{discrete-time METANET update map} around an equilibrium state $(\rho_e, v_e)$ \ch{under the assumption of small-signal perturbations,} to yield the local state Jacobian matrix $\mathbf{A}$ and the upstream state Jacobian matrix $\mathbf{B}$. 
\ch{The linearization is performed about a spatially uniform equilibrium defined by \(\rho_i=\rho_{\mathrm{e}}\) and \(v_i=v_{\mathrm{eq}}=V(\rho_{\mathrm{e}})\) for all cells \(i\), with \(\rho_{\mathrm{e}}=35\) chosen to represent a congested regime just above the critical density (\(\rho_{\mathrm{cr}}=30\)).} \sr{The results should thus be interpreted as specific to these operating conditions, rather than characterizing the model's stability across all traffic conditions.}

Matrix $\mathbf{A}$ captures the effect of the segment's current state on its own future state. It contains the partial derivatives of the state update equations with respect to the local segment variables:
\begin{equation} \label{eq:jacobian_A}
\mathbf{A} = 
\begin{bmatrix}
1 - \frac{T}{L}v_e & -\frac{T}{L}\rho_e \\[1em]
\frac{T}{\tau}V'(\rho_e) + \frac{\nu T}{\tau L (\rho_e + \kappa)} & 1 - \frac{T}{\tau} - \frac{T}{L}v_e
\end{bmatrix}
\end{equation}

\begin{figure}[t]
    \centering
    \includegraphics[width=0.45\textwidth]{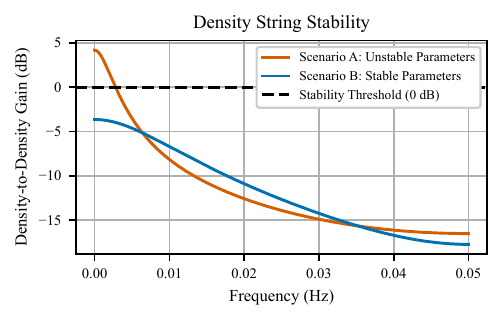}
    \caption{Density string stability in the frequency domain. The density-to-density gain of upstream boundary disturbances for two different realistic parameter sets. Gain values exceeding the 0 dB threshold indicate spatial amplification of the perturbation as it propagates to downstream segments.}
    \label{fig:string_stability}
\end{figure}

Matrix $\mathbf{B}$ captures the effect of the upstream segment's state on the local segment's future state. It contains the partial derivatives with respect to the upstream segment variables. Because the local speed update $v_i(k+1)$ in the METANET momentum equation does not depend directly on the upstream density $\rho_{i-1}(k)$, the term in the lower-left corner evaluates to zero:
\begin{equation} \label{eq:jacobian_B}
\mathbf{B} = 
\begin{bmatrix}
\frac{T}{L}v_e & \frac{T}{L}\rho_e \\[1em]
0 & \frac{T}{L}v_e
\end{bmatrix}
\end{equation}

\ch{Transforming the linearized state-space system defined by Jacobians \eqref{eq:jacobian_A} and \eqref{eq:jacobian_B} into the frequency domain via the $z$-transform yields the spatial transfer function:}
\begin{equation}
G(z) = (zI - \mathbf{A})^{-1}\mathbf{B}.
\end{equation}

This formulation addresses the boundary sensitivity problem. \ch{By extracting the density-to-density components of this transfer function, we isolate the system's response to perturbations in vehicle inflow (e.g., due to sensor noise).}
\ch{Note that while METANET is bidirectional due to the downstream anticipation term containing $\rho_t(l+1)$, this transfer function isolates the downstream propagation of upstream boundary noise to evaluate string stability.} 
Analyzing the frequency response of this specific subsystem allows us to compute the spatial gain of boundary noise. When the magnitude of this gain exceeds 0 dB at any frequency, it provides formal mathematical proof that boundary disturbances undergo amplification. 

Figure \ref{fig:string_stability} illustrates this frequency-dependent spatial amplification for two different sets of \ch{values from standard parameter bounds} \cite{zhao2025bounded}. As shown, the stability characteristics diverge drastically at low frequencies, which correspond to long-period macroscopic traffic waves. Scenario A (the more unstable configuration) exhibits a positive gain -- peaking near 4 dB as the frequency approaches 0 Hz -- indicating that it exponentially amplifies low-frequency boundary noise as it propagates downstream. Conversely, Scenario B remains strictly below the 0 dB threshold across the entire frequency spectrum, demonstrating robust string stability where boundary perturbations are continuously attenuated. Scenario A's massive amplification (a 4dB amplification means the disturbance grows by $\sim$58.5\% per segment) of low-frequency waves ultimately dominates, \ch{forcing the state to deviate from the nominal baseline}. Consequently, this approach explicitly links the foundational parameters of the METANET model to the empirical observation that boundary noise can exponentially propagate downstream, driving the irrecoverable divergence between perturbed and unperturbed simulations. 

This is particularly relevant for \ch{calibration with} real-world \ch{data}; empirical calibrations \ch{can} yield heterogeneous parameter sets where specific, highly unstable segments act as spatial ``triggers,'' cascading localized boundary noise into catastrophic, corridor-wide simulation divergence~\cite{raghavan2026dynamic}. 

\section{Analysis: Advantage of Dynamic Calibration}
\label{sec:analysis}

\sr{Section~\ref{sec:string_instability} identifies an amplification mechanism underlying METANET's input sensitivity; this section quantifies its consequence via a finite-horizon cost-deviation bound. We first show that the cost advantage from any static parameter vector to optimal dynamic parameters scales with the temporal variance of the ground-truth state. We then show that under small input perturbations, dynamic parameters yield a tighter cost-deviation bound, driven by the lower nominal cost.}

\begin{remark}
    \label{rem:perfdiff}
The total cost of the optimal dynamic parameters will be less than or equal to any static parameter vector $\bar \theta \in \mathcal{B}$:
\begin{equation}
    J(\bar\theta) - J(\theta_{0:H-1}^*) \geq 0
\end{equation}
\end{remark}

Applying the Performance Difference Lemma~\cite{kakade2002approximately} 
to the static parameter $\bar\theta$ and the optimal time-varying parameters 
$\theta^*_{0:H-1}$, the total cost gap decomposes as
\begin{equation}
    J(\bar\theta) - J(\theta_{0:H-1}^*) 
    = \sum_{t=0}^{H-1} A_t(x_t^{\bar\theta},\, \bar\theta),
\end{equation}
where $x_t^{\bar\theta}$ is the state trajectory induced by $\bar\theta$ 
starting at $x_0^{\rm gt}$. The result follows immediately since 
$A_t(x, \bar\theta) \geq 0$ for all $x, \bar\theta$ by 
definition~\eqref{eq:advantage}.

\subsection{Effect of Ground Truth Variance}

To establish the relationship between the advantage of dynamic calibration and the temporal variance of the ground truth traffic state, we first introduce two standard regularity assumptions regarding the cost landscape and the sensitivity of the model parameters.

\begin{assumption}[Local Quadratic Growth and Saturation]
\label{assum:convexity}
For any time step $t$ and state $x$, the action-value function $Q_t(x, \theta)$ exhibits local quadratic growth away from the optimum with a curvature parameter $\mu > 0$, bounded by a maximum cost discrepancy $D > 0$. This implies that for any parameter vector $\theta$ and the optimal parameter $\theta^*_t(x) = \argmin_{\theta'} Q_t(x, \theta')$, the following lower bound holds:
\begin{equation}
    Q_t(x, \theta) - Q_t(x, \theta^*_t(x)) \geq \min \left( \frac{\mu}{2} \|\theta - \theta^*_t(x)\|^2, D \right).
\end{equation}
\end{assumption}

\begin{assumption}[Parameter-State Sensitivity]
\label{assum:sensitivity}
For the trajectory of states $x_t^{\bar{\theta}}$ induced by a static parameter $\bar{\theta}$, there exists a sensitivity constant $L > 0$ such that the deviation of the optimal dynamic parameter at time $t$, $\theta^*_t(x_t^{\bar{\theta}})$, from the static parameter is bounded by the deviation of the ground truth state from its temporal mean $\bar{x}^{gt} = \frac{1}{H}\sum_{\tau=0}^{H-1} x_\tau^{gt}$:
\begin{equation}
    \|\bar{\theta} - \theta^*_t(x_t^{\bar{\theta}})\|^2 \geq L^2 \|x_t^{gt} - \bar{x}^{gt}\|^2.
\end{equation}
\end{assumption}

\ch{While Assumption 2 is a strong structural condition, it formalizes a} core intuition of macroscopic traffic simulation (e.g., METANET): a single static parameter \sr{will struggle to} simultaneously represent highly disparate traffic regimes (like free-flow and severe congestion). \sr{When the ground truth is exactly representable by a static parameter, the sensitivity constant $L$ collapses to 0, and Proposition 1 below reduces to Remark 1's bound, where the optimal time-varying parameters show no advantage. However, with real-world traffic, where METANET is modeling more complex dynamics, this assumption is likely satisfied}. Furthermore, Assumption \ref{assum:convexity} realistically models that extreme parameter divergences plateau in their penalization once traffic states reach gridlock. With these assumptions, we can formally bound the advantage of dynamic calibration relative to the variance of the ground truth state.

\begin{proposition} \label{prop:var}
    Denote the truncated temporal variance of the ground truth traffic state over the horizon $T$ as $\text{Var}_{D}(x^{gt}) = \frac{1}{H} \sum_{t=0}^{H-1} \min \left( \|x_t^{gt} - \bar{x}^{gt}\|^2, \frac{2D}{\mu L^2} \right)$. Under Assumptions 1 and 2, the total expected performance advantage of the optimal time-varying parameters $\theta_{0:H-1}^*$ over any static parameter $\bar{\theta}$ is lower-bounded by:
\begin{equation}
     J(\bar{\theta}) - J(\theta_{0:H-1}^*) \geq \frac{\mu L^2 H}{2} \text{Var}_{D}(x^{gt}).
\end{equation}
\end{proposition}

\begin{proof} 
From Remark~\ref{rem:perfdiff}, we established that the performance difference between a static parameter vector and the optimal dynamic parameters is exactly the cumulative sum of the advantage function along the state trajectory induced by $\bar{\theta}$:
\begin{equation}
       J(\bar{\theta}) - J(\theta_{0:H-1}^*) = \sum_{t=0}^{H-1} A_t(x_t^{\bar{\theta}}, \bar{\theta}) 
\end{equation}
By definition, the advantage function is the difference between the $Q$-value of the static parameter and the optimal value function:
\begin{equation}
    A_t(x_t^{\bar{\theta}}, \bar{\theta}) = Q_t(x_t^{\bar{\theta}}, \bar{\theta}) - V_t^*(x_t^{\bar{\theta}}).
\end{equation}
Since $V_t^*(x_t^{\bar{\theta}}) = \min_{\theta} Q_t(x_t^{\bar{\theta}}, \theta) = Q_t(x_t^{\bar{\theta}}, \theta_t^*(x_t^{\bar{\theta}}))$, we can rewrite the advantage function as:
\begin{equation}
    A_t(x_t^{\bar{\theta}}, \bar{\theta}) = Q_t(x_t^{\bar{\theta}}, \bar{\theta}) - Q_t(x_t^{\bar{\theta}}, \theta_t^*(x_t^{\bar{\theta}})).
\end{equation}
Applying Assumption \ref{assum:convexity} (Local Quadratic Growth), we lower-bound the instantaneous advantage function:
\begin{equation}
    A_t(x_t^{\bar{\theta}}, \bar{\theta}) \geq \min \left( \frac{\mu}{2} \|\bar{\theta} - \theta_t^*(x_t^{\bar{\theta}})\|^2, D \right).
\end{equation}
Next, we apply Assumption \ref{assum:sensitivity} (Parameter-State Sensitivity) to relate the parameter deviations to the ground truth traffic states. Since the $\min$ function is monotonically increasing with respect to its first argument, we have:
\begin{equation}
    A_t(x_t^{\bar{\theta}}, \bar{\theta}) \geq \min \left( \frac{\mu L^2}{2} \|x_t^{gt} - \bar{x}^{gt}\|^2, D \right).
\end{equation}
Summing this lower bound over the entire time horizon $H$ yields:
\begin{equation}
    \sum_{t=0}^{H-1} A_t(x_t^{\bar{\theta}}, \bar{\theta}) \geq \sum_{t=0}^{H-1} \min \left( \frac{\mu L^2}{2} \|x_t^{gt} - \bar{x}^{gt}\|^2, D \right).
\end{equation}
By factoring out $\frac{\mu L^2}{2}$, we can rewrite the summation in terms of the truncated variance $\text{Var}_{D}(x^{gt})$:
\begin{align}
    \sum_{t=0}^{H-1} A_t(x_t^{\bar{\theta}}, \bar{\theta}) &\geq \frac{\mu L^2 H}{2} \left[ \frac{1}{H} \sum_{t=0}^{H-1} \min \left( \|x_t^{gt} - \bar{x}^{gt}\|^2, \frac{2D}{\mu L^2} \right) \right] \nonumber \\
    &= \frac{\mu L^2 H}{2} \text{Var}_{D}(x^{gt}). 
\end{align}
Therefore, $J(\bar{\theta}) - J(\theta_{0:H-1}^*) \geq \frac{\mu L^2 H}{2} \text{Var}_{D}(x^{gt})$. This demonstrates that the suboptimality of using a static parameter $\bar{\theta}$ scales with the ground truth variance, but plateaus at extreme anomalies where parameter adjustments lose their marginal utility. \qed
\end{proof}

\subsection{Robustness to Data Perturbation}
Now, we derive an upper bound for the cost deviation when a parameter set $\theta$ is evaluated on perturbed boundary conditions $\tilde{d_t}$. This will analytically capture how a calibrated model reacts to off-nominal perturbations or noise. We introduce two assumptions to help formalize the dynamics of the state deviation. 

\begin{assumption}[First-Order Perturbation Approximation]
\label{ass:linearization}
The perturbation $\varepsilon$ is sufficiently small such that the 
state deviation evolves according to the first-order Taylor Expansion of $f$ around the  nominal point $(x_t, d_t)$:
\begin{equation}
    e_{t+1} = A_t e_t + B_t \varepsilon_t,
\end{equation}
where $A_t = \frac{\partial f}{\partial x}\big|_{(x_t,\, d_t;\, \theta_t)} 
\in \mathbb{R}^{2N \times 2N}$ is the state Jacobian and 
$B_t = \frac{\partial f}{\partial d}\big|_{(x_t,\, d_t;\, \theta_t)} \in \mathbb{R}^{2N \times m}$ is the 
boundary condition Jacobian, both evaluated along the nominal trajectory.
\end{assumption}

\sr{Note that the Jacobians $A_t$ and $B_t$ are different from matrices $\mathbf{A}$ and $\mathbf{B}$ in Section~\ref{sec:string_instability}. The latter are the equilibrium spatial Jacobians of a single segment update with respect to its own and its upstream neighbor's state, while $A_t$ and $B_t$ are the temporal Jacobians evaluated along the nominal trajectory $(x_t, d_t; \theta_t)$.}

\begin{assumption}[Bounded Jacobians]
\label{ass:bounded}
Define $\Phi_{j,i} := A_{j-1} A_{j-2} \cdots A_i$ for $i < j$ and $\Phi_{i,i} := \mathbf{I}$. There exist finite constants $M, \gamma > 0$ such that the product of state Jacobians satisfies
\begin{equation}
    \|\Phi_{j, i}\| \leq M \qquad \forall\; i, j \in [0, H],
\end{equation}
and the boundary condition Jacobian satisfies
\begin{equation}
    \|B_t\| \leq \gamma \qquad \forall\; t.
\end{equation}
\end{assumption}

Note that while $B_t$ is indexed by time in Assumption \ref{ass:bounded}, for METANET, $B_t$ is time invariant because it has the block structure
\begin{equation}
    B = \begin{pmatrix} 0_{N \times m} \\ B^\rho \end{pmatrix},
\end{equation}
where the upper block is zero since demand inputs do not appear directly 
in the speed equation, and the lower block $B^\rho \in \mathbb{R}^{N \times m}$ 
is a sparse matrix with value $\frac{T}{L\lambda_i}$ at $(i, j)$ if entry point $j$ is at segment $i$.
Since $B$ depends only on the fixed network geometry and simulation specifications 
and not on the traffic state or calibration parameters, it is 
time-invariant: $B_t = B$ for all $t$. The bound from Assumption \ref{ass:bounded}
is therefore
\begin{equation}
    \gamma = \max_{i \,\in\, \text{inflow entry}} 
    \frac{T}{L\lambda_i}.
\end{equation}

\begin{lemma}[State deviation bound]
\label{lem:state_deviation}
Under Assumptions~\ref{ass:linearization} and~\ref{ass:bounded}, the deviation between the nominal state $x_t$ and the perturbed state $\tilde{x}_t$ generated from boundary conditions $\varepsilon \in U$ satisfies
\begin{equation}
    \|e_t\| \leq \gamma \bar{\varepsilon} M t \qquad \forall\; t = 1, \ldots, H.
\end{equation}
\end{lemma}

\begin{proof}
Under Assumption~\ref{ass:linearization}, we unroll the linear recursion from 
$e_0 = 0$ using induction and show that $ e_t = \sum_{k=1}^{t} \Phi_{t,k}\, B_{k-1}\, \varepsilon_{k-1}$. For the base case:
\begin{align}
    e_1 &= A_0 \underbrace{e_0}_{=\,0} + B_0 \varepsilon_0 = B_0 \varepsilon_0.
\end{align}
Then, we assume that the statement holds for $e_t$ and show that it is true
for $e_{t+1}$. Using the linear recursion,
\begin{align}
    e_{t+1} &= A_t e_t + B_t \varepsilon_t \nonumber \\
            &= A_t \left( \sum_{k=1}^{t} \Phi_{t,k}\, B_{k-1}\, \varepsilon_{k-1} \right) + B_t \varepsilon_t \nonumber \\
            &= \sum_{k=1}^{t} \underbrace{A_t \Phi_{t,k}}_{=\, \Phi_{t+1,k}}\, B_{k-1}\, \varepsilon_{k-1} \;+\; \underbrace{\Phi_{t+1,\,t+1}}_{=\, I}\, B_t \varepsilon_t \nonumber \\
            &= \sum_{k=1}^{t+1} \Phi_{t+1,k}\, B_{k-1}\, \varepsilon_{k-1}.
\end{align}
Therefore, by induction, $ e_t = \sum_{k=1}^{t} \Phi_{t,k}\, B_{k-1}\, \varepsilon_{k-1}$. Taking norms and applying the triangle inequality,
\begin{align}
    \|e_{t}\| &\leq \sum_{k=1}^{t} \|\Phi_{t,k}\|\, \|B_{k-1}\|\, \|\varepsilon_{k-1}\| \nonumber \\
                &\leq \sum_{k=1}^{t} M \cdot \gamma \cdot \bar{\varepsilon} \nonumber \\
                &=  \gamma \bar{\varepsilon} M t \qquad \forall\; t = 1, \ldots, H,
\end{align}
where the second inequality applies Assumption~\ref{ass:bounded} ($\|\Phi_{j, i}\| \leq M$, $\|B_{k-1}\| \leq \gamma$) and the uncertainty bound ($\|\varepsilon_{k-1}\| \leq \bar{\varepsilon}$). \qed
\end{proof}

For convenience, define the shorthand
\begin{equation}
    G_t := \gamma \bar{\varepsilon} M t,
\end{equation}
so that $\|e_t\| \leq G_t$ by Lemma~\ref{lem:state_deviation}. We additionally define the nominal residual of parameter $\theta_t$ at stage $t$ as
\begin{equation}
    R_t(\theta_t) := \|x_t(\theta_t) - x_t^{gt}\|,
\end{equation}
which, similarly to the cost function, measures the error of the nominal state induced by a parameter vector $\theta_t$ with respect to the ground truth.

\begin{theorem}[Total cost deviation bound]
\label{thm:cost_deviation}
Under Assumptions~\ref{ass:linearization} and~\ref{ass:bounded}, the deviation 
in total cost under a perturbation $\varepsilon \in U$ satisfies
\begin{multline} \label{eq:cost_deviation_bound}
    \left| J_\varepsilon(\theta_{0:H-1}) - J(\theta_{0:H-1}) \right| \\
    \leq \sum_{t=1}^{H} G_t^2 + 2\sqrt{J(\theta_{0:H-1})} \cdot \sqrt{\sum_{t=1}^{H} G_t^2},
\end{multline}
where $G_t = \gamma \bar{\varepsilon} M t$ and $ J_\varepsilon(\theta_{0:H-1}) = \sum_{t=0}^{H} 
c(\tilde{x}_t, \theta_t)$.
\end{theorem}

\begin{proof}
Recall that the total nominal cost is $J(\theta_{0:H-1}) = \sum_{t=0}^{H} c(x_t, \theta_t)$. The difference between the perturbed and nominal cost is
\begin{align}
    J_\varepsilon(\theta_{0:H-1}) - J(\theta_{0:H-1})
    &= \sum_{t=0}^{H} \left[ c(\tilde{x}_t, \theta_t) - c(x_t, \theta_t) \right].
\end{align}
Taking absolute values and applying the triangle inequality:
\begin{align}
    \left| J_\varepsilon(\theta_{0:H-1}) - J(\theta_{0:H-1}) \right|
    &\leq \sum_{t=0}^{H} \left| c(\tilde{x}_t, \theta_t) - c(x_t, \theta_t) \right|.
\end{align}
Since $e_0 = 0$ by assumption, the $t = 0$ term vanishes. For $t \geq 1$, 
we can substitute $\tilde{x}_t = x_t + e_t$ into the stage cost 
$c(x_t, \theta_t) = \|x_t - x_t^{gt}\|^2$:
\begin{align}
    c(\tilde{x}_t, \theta_t) - c(x_t, \theta_t)
     &= \|\tilde{x}_t - x_t^{gt} \|^2 - \|x_t - x_t^{gt}\|^2 \nonumber \\
    &= \|(x_t - x_t^{gt}) + e_t\|^2 - \|x_t - x_t^{gt}\|^2 \nonumber \\
    &= \|e_t\|^2 + 2(x_t - x_t^{gt})^T e_t.
\end{align}
Taking absolute values and applying the Cauchy-Schwarz inequality:
\begin{align}
    \left| c(\tilde{x}_t, \theta_t) - c(x_t, \theta_t) \right|
    &\leq \|e_t\|^2 + 2\, \|x_t - x_t^{gt}\|\, \|e_t\| \nonumber \\
    &= \|e_t\|^2 + 2\, R_t(\theta_t)\, \|e_t\|.
\end{align}
Substituting the state deviation bound $\|e_t\| \leq G_t$ from 
Lemma~\ref{lem:state_deviation} and summing over $t = 1, \ldots, T$:
\begin{align}
    \left| J_\varepsilon(\theta_{0:H-1}) - J(\theta_{0:H-1}) \right|
    &\leq \sum_{t=1}^{H} \left( \|e_t\|^2 + 2\, R_t(\theta_t)\, \|e_t\| \right) \notag \\
    &\leq \sum_{t=1}^{H} \left( G_t^2 + 2\, R_t(\theta_t)\, G_t \right). 
\end{align}

\begin{figure}[tbp]
    \centering
    \includegraphics[width=0.4\textwidth]{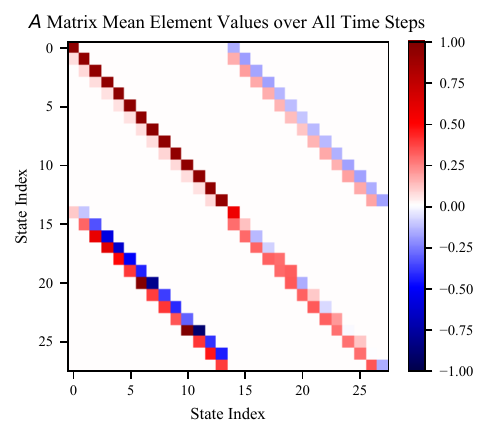}
    \caption{Mean A matrix for the synthetic scenario.}
    \label{fig:mean_a_matrix}
\end{figure}

Note that $\sum_{t=1}^{H} R_t(\theta_t) \, G_t$ can be rewritten in terms of 
the cost function as follows. By the Cauchy-Schwarz inequality:
\begin{align}
    \sum_{t=1}^{H} R_t(\theta_t) \, G_t
    &\leq \sqrt{\sum_{t=1}^{H} R_t(\theta_t)^2} \cdot \sqrt{\sum_{t=1}^{H} G_t^2} \nonumber \\
    &= \sqrt{J(\theta_{0:H-1})} \cdot \sqrt{\sum_{t=1}^{H} G_t^2},
\end{align}
Substituting back into the total cost deviation bound:
\begin{align}
    \left| J_\varepsilon(\theta_{0:H-1}) - J(\theta_{0:H-1}) \right|
    &\leq \sum_{t=1}^{H} G_t^2 + \nonumber \\ 
    & 2\sqrt{J(\theta_{0:H-1})} \cdot 
      \sqrt{\sum_{t=1}^{H} G_t^2} 
\end{align}
%

\qed
\end{proof}

The total cost deviation bound decomposes cleanly into two terms at each stage. The quadratic term 
$G_t^2$ captures the effect of the perturbation amplified through the 
dynamics, growing as $M^{2}$ and linearly with $t$. 

The first term will dominate for large $M$ and for large $\bar \varepsilon$. Note that $M$ is the maximum norm of the product of state Jacobian matrices, $A_t$. However, due to the sparsity and block \ch{tridiagonal} nature of $A_t$, which is inherent to the structure of METANET, the transition matrices $\Phi_{j,i}$ remain 
well-bounded in practice. Figure \ref{fig:mean_a_matrix} visualizes this structure for the mean of all $A_t$ \ch{from Section~\ref{subsec:casestudy_ii}}; in Section~\ref{case_study_syn} we show that $M$ remains fairly small for both $\theta^*$ and $\bar\theta$.

\sr{The second term captures the interaction between the nominal calibration error and the perturbations, showing that parameters with smaller nominal cost $ J(\theta_{0:H-1})$ yield a tighter bound. This, when combined with Remark \ref{rem:perfdiff} which shows that $J(\theta^*) < J(\bar \theta)$, confirms that in the small-perturbation regime, the optimal time-varying parameters $\theta^*$ yield a tighter cost-deviation upper bound than any static parameter $\bar \theta$.}

\section{Experiments} 
\label{sec:experiments}

\subsection{Methodology}
Per the method proposed in \cite{raghavan2026dynamic}, we solve the optimal control problem in equation \ref{eq:calibration-ocp} using Model Predictive Control. This provides an approximate solution for the optimal dynamic parameters $\theta^*$. For the static parameters, the control horizon is set to the time horizon of the entire simulation ($H$), and for the dynamic parameters, the control horizon is set to the length of the simulation timestep, $T$, with a prediction horizon of $6T$. For this work, we use $T=10$ seconds and $L = 400 $ meters to sufficiently satisfy the Courant-Friedrich Condition. IPOPT 3.14.16 was used as the optimization solver; it is known to demonstrate effectiveness for large-scale nonlinear optimizations problems and is thus commonly used in these settings \cite{biegler2009large}.

\subsection{Case Study I: Synthetic Scenario} \label{case_study_syn}

\begin{figure}[t]
    \centering
    \includegraphics[width=0.45\textwidth]{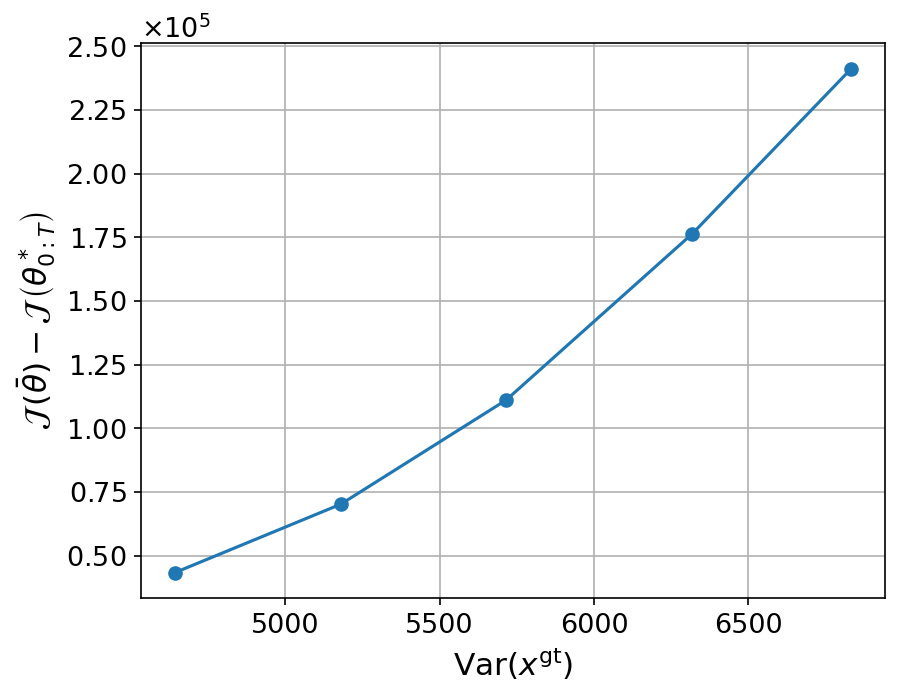}
    \caption{Advantage of dynamic parameters compared to static parameters vs variance of ground truth state. Note: $\sigma$ is mapped to the variance of the state itself, which is shown on the x-axis.}
    \label{fig:var_gt}
\end{figure}

To evaluate the validity of our variance and robustness bounds, we construct a  1-hour synthetic traffic scenario 
using the METANET model with a known ground truth parameter 
set from \cite{chavoshi2023feedback}. The scenario models a bottleneck on a 6 km highway, where the first 4 km have 4 lanes after which it drops down to 2 lanes. The upstream inflow spikes to 5500 veh/hr between minutes 5 to 20, outside of which it remains under capacity at 4000 veh/hr. All parameters are held fixed at their nominal values across the 
horizon, with the exception of the free-flow speed $v_{free}$, which is 
varied over time to simulate a recovering traffic condition, for 
example, the dissipation of an incident or the transition out of a 
morning peak period.

Specifically, $v_{free}$ is initialized at a reduced value of 
$v_{free}(0) = v_{free}^{\mathrm{nom}} \cdot (1 - \sigma)$, where 
$\sigma \in (0,1)$ is a variance parameter controlling the severity of 
the initial reduction. Over the first 45 minutes of the simulation 
horizon, $v_{free}$ increases linearly from this reduced value to the 
nominal free-flow speed of $120$ km/h, after which it remains constant 
for the final 15 minutes. The resulting ground truth trajectory 
$\{x_t^{gt}\}_{t=0}^{T}$ is then generated by simulating the METANET 
model forward using these time-varying parameters, and both the static 
and dynamic parameters are determined by calibrating to this data. The advantage of $\theta^*$ over $\bar \theta$ for different $\sigma$ values is shown in Figure \ref{fig:var_gt}, \sr{and the x-axis reports the variance of the ground truth data for each $\sigma$. The increasing trend across $\sigma$ empirically supports Proposition \ref{prop:var}, which bounds the dynamic parameter advantage from below by a quantity that grows with the temporal variance of the ground truth state.}

\begin{figure}[t]
    \centering
    \includegraphics[width=0.47\textwidth]{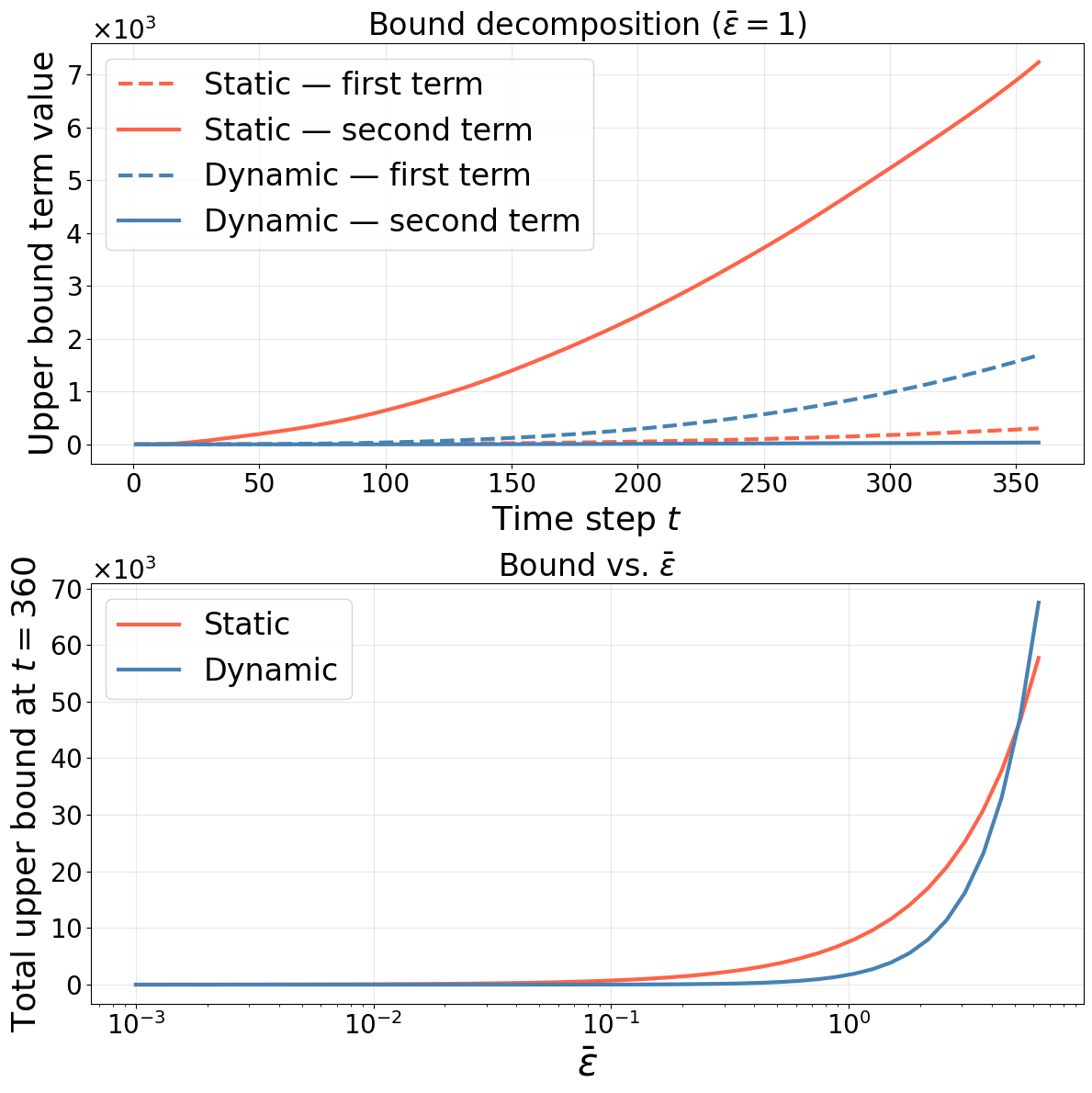}
    \caption{(Top) Decomposition of the cost deviation bound; (Bottom)  Total cost 
deviation bound $\sum_t G_t^2 + 2\sqrt{J(\theta_{0:T}) \, \sum_t G_t^2}$ 
evaluated at $t = H$ as a function of perturbation size $\bar{\varepsilon}$ }
    \label{fig:rob_ub}
\end{figure}

Now we turn to the robustness bound from Theorem \ref{thm:cost_deviation}. The resulting values of $M$ for the static and dynamic parameters are 2.55 and 6.04, respectively. At first glance, the larger value of $M$ for the dynamic parameters may appear to suggest that the time-varying calibration is less robust. However, $M$ alone does not determine the tightness of the bound. In fact, the top plot of Figure \ref{fig:rob_ub} shows that when decomposing the right hand side (RHS) of the robustness upper bound in \eqref{eq:cost_deviation_bound} for $\bar\varepsilon = 1$, the second term is the one that dominates throughout the horizon, emphasizing the larger role of the nominal calibration error $J(\theta^*)$ in bounding the perturbation cost deviation. The full RHS of \ch{the} cost deviation bound is plotted in the bottom panel of Figure~\ref{fig:rob_ub} for both the static and dynamic parameters across a range of perturbation sizes $\bar{\varepsilon}$. The bound is tighter for the dynamic parameters across most values of $\bar{\varepsilon}$, driven by the lower nominal cost $J(\theta^*) \leq J(\bar \theta)$. However, we see that the static parameter bound overtakes that of the dynamic parameters when $\bar\epsilon$ becomes too large, likely indicating that Assumption $\ref{ass:linearization}$ is breaking. \sr{This confirms the implication of Theorem~\ref{thm:cost_deviation}: in the small-perturbation regime, models with lower nominal cost yield a tighter cost-deviation upper bound, indicating reduced sensitivity to these boundary condition perturbations.}



\subsection{Case Study II: I-24 MOTION Testbed}
\label{subsec:casestudy_ii}

Aligned with \cite{raghavan2026dynamic}, we use traffic data for the westbound corridor in the I-24 MOTION INCEPTION v1.0.0 dataset. The scenario used extends from 8am to 9am on Wednesday, Nov. 30th, 2022. The source of the data is 276 cameras over 4 miles of interstate I-24 in Tennessee, USA. This footage is processed into vehicle trajectory data \cite{gloudemans202324}, which we then convert into macroscopic values based on Edie’s definition and adjust to the level of detail required by our METANET model. We use $L = 0.4$ kilometers and $T = 10$ seconds.
Parameters were tuned for 14 segments, covering 5.6 km of highway. We selected the data because of the stop-and-go traffic waves present on I-24 in this rush hour scenario. To compensate for sensor inaccuracies, boundary conditions were smoothed.

\sr{Figure \ref{fig:rob_ubi24} shows that the cost-deviation bound remains tighter for the dynamic parameters than for the static parameter vector across the range of $\bar\varepsilon$ considered, mirroring the synthetic result in Figure~\ref{fig:rob_ub}. This is an indication that Theorem~\ref{thm:cost_deviation} extends beyond the setting in which its assumptions are exactly satisfied, since real I-24 data may violate the first-order approximation in Assumption~\ref{ass:linearization}.}

\begin{figure}[tbp]
    \centering
    \includegraphics[width=0.47\textwidth]{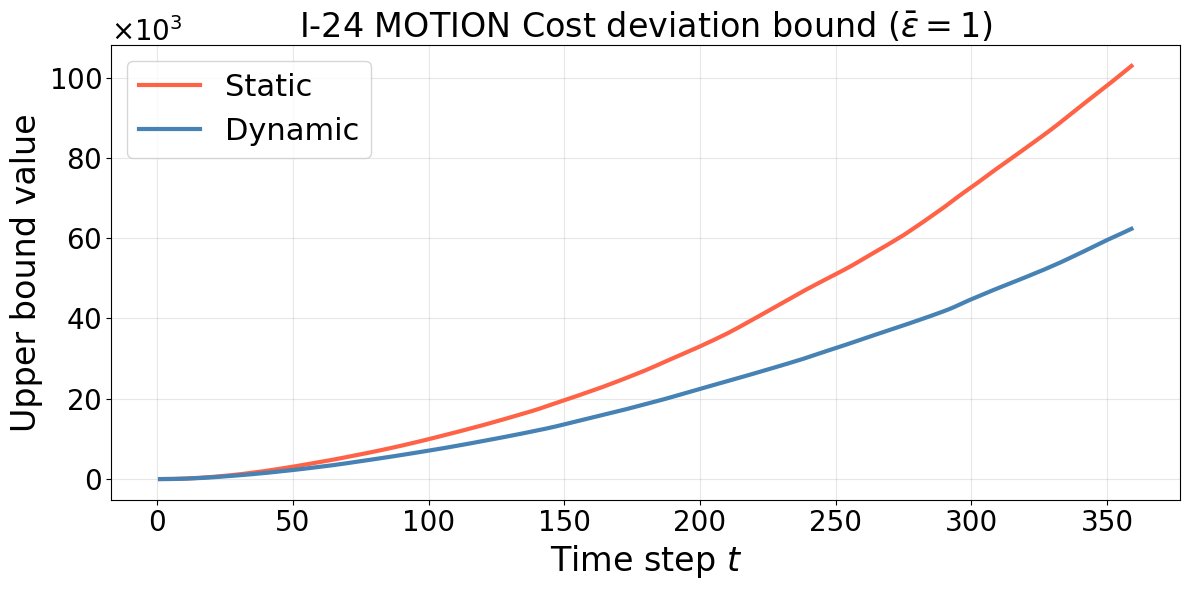}
    \caption{Cost deviation upper bound for dynamic vs static parameters over the 1-hour time horizon for I-24 MOTION data.}
    \label{fig:rob_ubi24}
\end{figure}

\section{Conclusion}
\label{sec:conclusion}
This work provides a theoretical explanation for the empirically observed 
robustness of dynamic METANET calibration~\cite{raghavan2026dynamic}. We formalized the sensitivity of calibrated models to exogenous boundary perturbations via string stability analysis, derived an analytical cost deviation bound that is provably tighter for optimal time-varying parameters than for any static parameters, and showed that the performance advantage of dynamic calibration grows monotonically with ground truth state variance. These findings were validated on both a synthetic scenario and real-world I-24 MOTION data, confirming that dynamic calibration yields models that are simultaneously more accurate and more robust to data perturbations.

\ch{Future work could investigate perturbations on other boundary inputs, expand the string stability analysis to account for time-varying inflows and equilibria induced by dynamic parameter updates, and develop improvements that concentrate the calibration \sr{optimization} process in more robust parameter spaces.}

\balance 
	\section*{ACKNOWLEDGMENTS}
	This work was done with the support of the NSF Graduate Research Fellowship under Grant No. 2141064. \ch{The authors acknowledge the use of large language models (LLMs) for assistance in copyediting and drafting. The authors have reviewed and revised all material generated with AI and accept full responsibility for the final content of this publication.}
	\bibliographystyle{IEEEtran}
	\bibliography{root} 
	
\end{document}